\documentclass[letterpaper, 10 pt]{IEEEtran} 
\usepackage{graphics} % for pdf, bitmapped graphics files
\usepackage{epsfig} % for postscript graphics files
\usepackage{mathptmx} % assumes new font selection scheme installed
\usepackage{times} % assumes new font selection scheme installed
\usepackage{amsmath} % assumes amsmath package installed
\usepackage{amssymb} % assumes amsmath package installed

\usepackage{amsthm}
\usepackage{amssymb}
\usepackage{xypic}
\theoremstyle{plain}

\newtheorem{thm}{Theorem}
\newtheorem{proposition}[thm]{Proposition}
\newtheorem{cor}[thm]{Corollary}
\newtheorem{lem}[thm]{Lemma}

\theoremstyle{definition}
\newtheorem{df}[thm]{Definition}

\newtheorem{eg}[thm]{Example}

\def\shf{\mathcal}

\def\st{\textrm{star }}
\def\cl{\textrm{cl }}
\def\active{\textrm{active}}
\def\roi{\textrm{roi }}

\title{Protocol-independent critical node detection}
\author{\IEEEauthorblockN{Michael Robinson, Jimmy Palladino}\\
\IEEEauthorblockA{Department of Mathematics and Statistics\\
American University\\
4400 Massachusetts Ave NW\\
Washington, DC 20016\\
Email: michaelr@american.edu}}

\begin{document}
\maketitle
\begin{abstract}
This article explains how to construct a protocol-independent model for passing traffic through a wireless network with a single channel carrier sense multiple access/collision detection (CSMA/CD) media access model.
\end{abstract}
%
%\begin{keywords}
%wireless network; simplicial complex; sheaf; CSMA/CD
%\end{keywords}
%
\section{Introduction}

When a carrier sense multiple access/collision detection (CSMA/CD) media access model is used in a wireless network, only one node in a given vicinity can transmit while the others must wait.  Although the physical layer protocols of wireless networks can be quite complex, the basic topology of the network plays an important role in determining network performance.  This paper addresses the problem of identifying critical nodes and links within a network by using local invariants derived from the local topology of the network.  Recognizing that although protocol plays an important role, we are specifically concerned with those effects that are \emph{protocol independent}.

This paper provides theoretical justification for the ``right'' local neighborhood in a wireless network with a CSMA/CD media access model using the structure of network activation patterns, and then validates the resulting topological invariants using simulated network traffic generated with $ns2$.

\section{Historical context and contributions}

Graph theory methods have been used extensively (for instance \cite{Nandagopal_2000,Yang_2002,Jain_2003,Lee_2007}) for identifying critical nodes in a network that carry a disproportionate amount of traffic.  However, direct application of graph theory to locate these nodes is computationally expensive \cite{DiSumma_2011,dinh2012}.  Furthermore, graphs are better suited to \emph{wired} networks and don't necessarily address the multi-way interactions inherent in wireless networks \cite{Chiang_2007}.

The present paper extends our previous work \cite{Joslyn_2016,RobinsonGlobalSIP2014} that used higher-dimensional abstract simplicial complexes instead of graphs and used connectivity as a measure of network health.  Although connectivity can be a useful measure of health \cite{Noubir_2004,Gueye_2010}, it is rather coarse.  We remedy this with a more systematic study of an 802.11b wireless network using the $ns2$ network simulator \cite{nsnam}.

Our previous work also used an apparently \emph{ad hoc} definition of the local neighborhood of a node in order to perform its analysis.  This article provides solid theoretical justification for that choice, and demonstrates the viability of the resulting \emph{local homological} vulnerability of a node.  This provides a faster method of identifying critical nodes than direct graph-theoretic ones.

\section{Interference from a transmission}

This paper advocates the use of abstract simplicial complexes (which generalize undirected graphs) as a means of modeling wireless networks. 

\begin{df}
An \emph{abstract simplicial complex} $X$ on a set $A$ is a collection of ordered subsets of $A$ that is closed under the operation of taking subsets.  We call an element of $X$ which itself contains $k+1$ elements a \emph{$k$-cell}.  We usually call a $0$-cell a \emph{vertex} or \emph{node} and a $1$-cell an \emph{edge}.  

If $a,b$ are cells with $a \subset b$, we say that $a$ is a \emph{face} of $b$, and that $b$ is a \emph{coface} of $a$.  A cell of $X$ that has no cofaces is called a \emph{facet}.

For a set $Y$ of cells in $X$, we let the \emph{closure} $(\cl Y)$ be the smallest abstract simplicial complex that contains $Y$ and the \emph{star} $(\st Y)$ be the set of all cells that have at least one face in $Y$.  
\end{df}

Let a wireless network consist of a single channel, with nodes $N=\{n_i\}$ in a region $R$.  Associate an open set $U_i \subset R$ to each node $n_i$ that represents its \emph{transmitter coverage region}.  For each node $n_i$, a continuous function $s_i : U_i \to \mathbb{R}$ represents its \emph{signal level} at each point in $U_i$.  Without loss of generality, we assume that there is a global threshold $T$ for accurately decoding the transmission from any node.  In \cite{RobinsonGlobalSIP2014}, two abstract simplicial complex models were developed: the \emph{interference} and \emph{link} complexes.

\begin{df}
The \emph{interference complex} $I=I(N,U,s,T)$ consists of all subsets of $N$ of the form $\{i_1, \dotsc, i_n\}$ for which $U_{i_1} \cap \dotsb \cap U_{i_n}$ contains a point $x \in R$ for which $s_{i_k}(x) > T$ for all $k=1, \dotsb n$. 
\end{df}
Briefly, the interference complex describes the lists of transmitters that when transmitting will result in at least one mobile receiver location receiving multiple signals simultaneously.  (The interference complex is a \v{C}ech complex \cite{Hatcher_2002}.)

\begin{proposition}
Each facet of the interference complex corresponds to a maximal collection of nodes that mutually interfere.
\end{proposition}
\begin{proof}
Let $c$ be a cell of the interference complex.  Then $c$ is a collection of nodes whose coverages have a nontrivial intersection.  The decoding threshold is exceeded for all nodes at some point $x$ in this intersection.  If any two nodes in $c$ transmit simultaneously, they will interfere at $x$.  If $c$ is a facet, it is contained in no larger cell, so it is clearly maximal.
\end{proof}

\begin{df}
The \emph{link graph} is the following collection of subsets of $N$:
\begin{enumerate}
\item $\{n_i\} \in N$ for each node $n_i$, and
\item $\{n_i,n_j\}\in N$ if $s_i(n_j) > T$ and $s_j(n_i) >T$.
\end{enumerate}
The \emph{link complex} $L=L(N,U,s,T)$ is the clique complex of the link graph, which means that it contains all elements of the form $\{i_1,\dotsc,i_n\}$ whenever this set is a clique in the link graph.
\end{df}

\begin{proposition}
Each facet in the link complex is a maximal set of nodes that can communicate directly with one another (with only one transmitting at a time).
\end{proposition}
\begin{proof}
Let $c$ be a cell of the link complex.  By definition, for each pair of nodes, $i,j\in c$ implies that $s_i(n_j) > T$ and $s_j(n_i) >T$.  Therefore, $i$ and $j$ can communicate with one another.  
\end{proof}

\begin{cor}
Facets of the link complexes represent common broadcast resources.
\end{cor}

Since the CSMA/CD protocol is implemented locally, it can be modeled as follows:

\begin{df}
Suppose that $X$ is a simplicial complex (such as an interference or link complex) whose set of vertices is $N$.  Consider the following assignment $\shf{A}$ of additional information to capture which nodes are transmitting and decodable:
\begin{enumerate}
\item To each cell $c\in X$, assign the set
\begin{eqnarray*}
\shf{A}(c)&=&\{n \in N : \text{there exists a cell } d\in X \text{ with }\\
&& c \subset d \text{ and } n\in d\}\cup\{\perp\}
\end{eqnarray*} 
of nodes that have a coface in common with $c$, along with the symbol $\perp$.  We call $\shf{A}(c)$ the \emph{stalk} of $\shf{A}$ at $c$.
\item To each pair $c \subset d$ of cells, assign the \emph{restriction function}
\begin{equation*}
\shf{A}(c\subset d)(n) =
\begin{cases}
n&\text{if }n\in \shf{A}(d)\\
\perp&\text{otherwise}\\
\end{cases}
\end{equation*}
\end{enumerate}
\end{df}

For instance, if $c \in X$ is a cell of a link complex, $\shf{A}(c)$ specifies which nearby node is transmitting and decodable, or $\perp$ if none are.  The restriction functions relate the decodable transmitting nodes at the nodes to which nodes are decodable along an attached wireless link.  Similarly, if $c \in X$ is a cell of an interference complex, $\shf{A}(c)$ also specifies which nearby node is transmitting, and effectively locks out any interfering transmissions from other nodes.  

\begin{df}
The assignment $\shf{A}$ is called the \emph{activation sheaf} and is an example of a \emph{cellular sheaf} \cite{Shepard_1985} -- a mathematical object that stores local data.  The theory of sheaves explains how to extract consistent information, which in the case of networks consists of nodes whose transmissions do not interfere with one another.

A \emph{section} of $\shf{A}$ supported on a subset $Y \subseteq X$ is a function $s:Y \to N$ so that for each $c \subset d$ in $Y$, $s(c) \in \shf{A}(c)$ and $\shf{A}(c\subset d)\left( s(c)\right) = s(d)$.  A section supported on $X$ is called a \emph{global section}.  
\end{df}

Specifically, global sections are complete lists of nodes that can be transmitting without interference.  

\begin{figure}
\begin{center}
\includegraphics[width=2.5in]{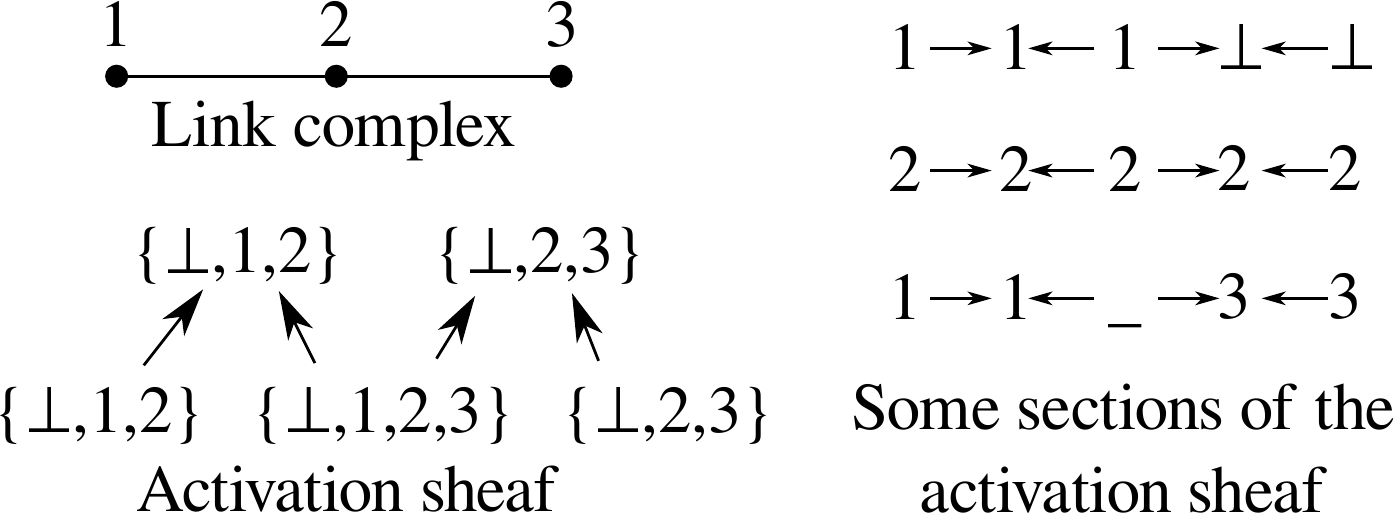}
\caption{A link complex (left top), sheaf $\shf{A}$ (left bottom), and three sections (right).  The restrictions are shown with arrows.  global section when node 1 transmits (right top), global section when node 2 transmits (right middle), and a local section with nodes 1 and 3 attempting to transmit, interfering at node 2 (right bottom)}
\label{fig:linesec}
\end{center}
\end{figure}

\begin{eg}
\label{eg:linesec}
Figure \ref{fig:linesec} shows a network with three nodes, labeled 1, 2, and 3.  When node 1 transmits, node 2 receives.  Because node 2 is busy, its link to node 3 must remain inactive (right top).  When node 2 transmits, both nodes 1 and 3 receive (right middle).  The right bottom diagram shows a local section that cannot be extended to the cell marked with a blank.  This corresponds to the situation where nodes 1 and 3 attempt to transmit but instead cause interference at node 2. 
\end{eg}

\begin{df}
Suppose that $s$ is a global section of $\shf{A}$.  The \emph{active region} associated to a node $n\in X$ in $s$ is the set
\begin{equation*}
\active (s,n) = \{a \in X : s(a)=n\},
\end{equation*}
which is the set of all nodes that are currently waiting on $n$ to finish transmitting.
\end{df}

\begin{lem}
\label{lem:act}
The active region of a node is a connected, closed subcomplex of $X$ that contains $n$.
\end{lem}
\begin{proof}
Consider a cell $c \in \active (s,n)$.  If $c$ is not a vertex, then there exists a $b \subset c$; we must show that $b\in \active (s,n)$.  Since $s$ is a global section $\shf{A}(b\subset c)s(b)=s(c)=n$.  Because $s(c) \not= \perp$, the definition of the restriction function $\shf{A}(b\subset c)$ implies that $s(b)=n$.  Thus $b\in \active (s,n)$ so $\active (s,n)$ is closed.

If $c \in \active (s,n)$, then $c$ and $n$ have a coface $d$ in common.  Since $s$ is a global section $s(d)=\shf{A}(c \subset d)s(c)=\shf{A}(c \subset d)n=n$.  Thus, $n \in \active(s,n)$, because $n$ is a face of $d$ and $\active (s,n)$ is closed.  This also shows that every cell in $\active (s,n)$ is connected to $n$.
\end{proof}

\begin{lem}
\label{lem:stactive}
The star over the active region of a node does not intersect the active region of any other node.
\end{lem}
\begin{proof}
Let $c\in \st \active (s,n)$.  Without loss of generality, assume that $c \notin \active (s,n)$.  Therefore, there is a $b \in \active (s,n)$ with $b\subset c$.  By the definition of the restriction function $\shf{A}(b\subset c)$, the assumption that $c\notin \active (s,n)$, and the fact that $s$ is a global section, $s(c)$ must be $\perp$.
\end{proof}

\begin{cor}
  If $s$ is a global section of an activation sheaf $\shf{A}$, then the \emph{support} of $s$ -- the set of cells $c$ where $s(c) \not= \perp$ -- consists of a disjoint union of active regions of nodes. 
\end{cor}

\begin{lem}
  \label{lem:activeinvariant}
The active region of a node is independent of the global section.  More precisely, if $r$ and $s$ are global sections of $\shf{A}$ and the active regions associated to $n \in X$ are nonempty in both, then $\active (s,n)=\active (r,n)$.
\end{lem}
\begin{proof}
Without loss of generality, we need only show that $\active (s,n) \subseteq \active (r,n)$.  If $c \in\active(s,n)$, there must be a cell $d\in X$ that has both $n$ and $c$ as faces.  Now $s(n)=r(n)=n$ by Lemma \ref{lem:act}, which means that $r(d)=\shf{A}(n \subset d)r(n)=n$.  Therefore, since $\active (r,n)$ is closed, this implies that $c \in \active(r,n)$.  
\end{proof}

\begin{cor}
  \label{cor:activationsections}
  The space of global sections of an activation sheaf consists of all sets of nodes that can be transmitting simultaneously without interference.
\end{cor}

\section{Using activation patterns}

The structure of the global sections of an activation sheaf leads to a model in which an active node silences all other nodes in its vicinity.  

\begin{df}
  Because of the Lemmas, we call the star over an active region associated to a node $n$ the \emph{region of influence}.  The region of influence of a facet is the star over the closure of that facet.  The region of influence for a collection of facets $F$ can be written as a union
\begin{equation*}
\roi F = \bigcup_{f \in F} \st \cl f.
\end{equation*}
\end{df}

In our previous work \cite{RobinsonGlobalSIP2014}, the region of influence was used without detailed justification; the following Corollary provides this needed justification.

\begin{cor}
\label{cor:unaffected}
The complement of the region of influence of a facet is a closed subcomplex.
\end{cor}

Given this justification, \cite{RobinsonGlobalSIP2014} shows that critical nodes or links are those cells $c$ for whom the local homology dimension (see also \cite{Joslyn_2016})
\begin{equation*}
  LH_k(c) = \text{dim } H_k(X,X\,\backslash\,\roi c)
\end{equation*}
is larger than the average.

This implies the following experimental hypothesis: \emph{If a node is critical, it will have a large local homology dimension}.  Since the $ns2$ network simulator provides complete transcripts of all packets, we can define a critical node to be one that \emph{forwards} a large number of packets compared to other nodes in the network \cite{Arulselvan_2009}.

\begin{figure}
\begin{center}
\includegraphics[width=2in]{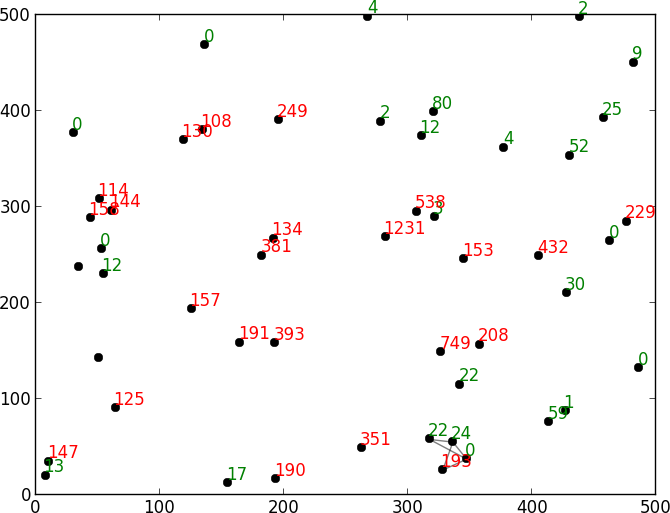} 
\caption{Locations of nodes and forwarded packet counts (axes in meters)}
\label{fig:network}
\end{center}
\end{figure}

We constructed a small simulation with 50 nodes as shown in Figure \ref{fig:network}.  Packets were randomly assigned source and destination nodes within the network, and all packet histories were recorded for analysis.

\begin{figure}
\begin{center}
\includegraphics[width=2in]{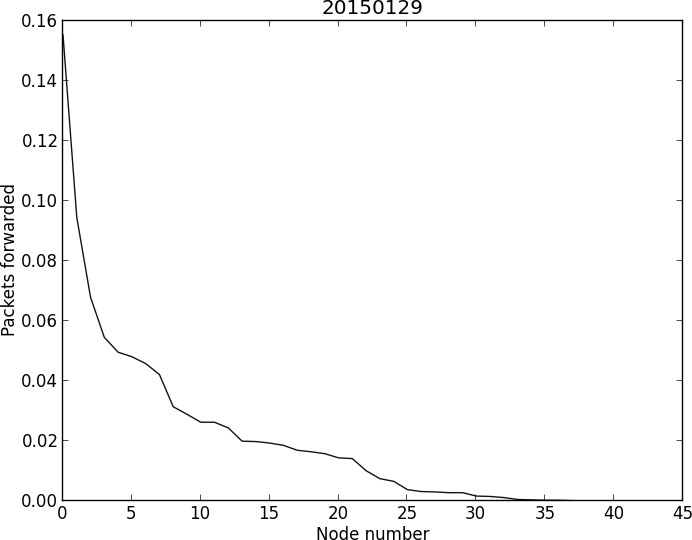} 
\caption{Probability that a given packet will be forwarded by a specific node}
\label{fig:fwdhist}
\end{center}
\end{figure}

Figure \ref{fig:fwdhist} shows the probability that a node will forward a random packet.  (The node numbers have been sorted from greatest to least probability.) The figure shows that most nodes forward only a small number of packets, while a few nodes carry considerably more traffic.

\begin{figure}
\begin{center}
\includegraphics[width=2in]{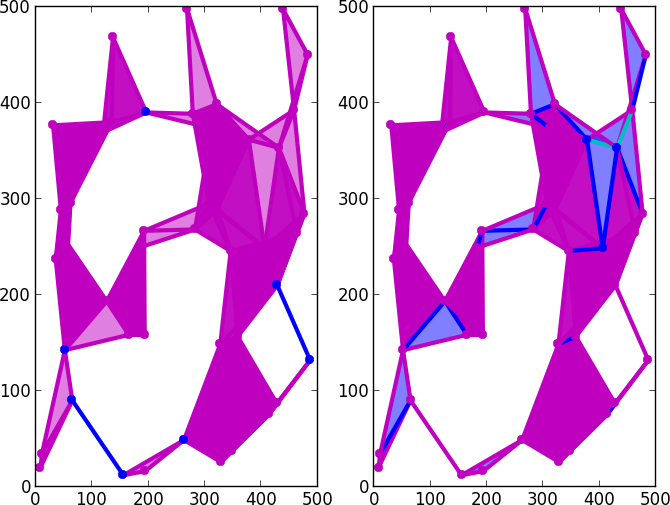} 
\caption{Dimension of local homology $LH_1$ (left) and $LH_2$ (right).  Axes in meters; Magenta = 0, Blue = 1, Cyan = 2.}
\label{fig:networklh}
\end{center}
\end{figure}

Figure \ref{fig:networklh} shows the dimension of local homology over all nodes and links in the network.  In this particular network, the local homology dimension is only 0, 1, or 2.  It is clear that nodes with high $LH_1$ occupy certain ``pinch points'' in the network.

\section{Conclusions}

\begin{figure}
\begin{center}
\includegraphics[width=3.5in]{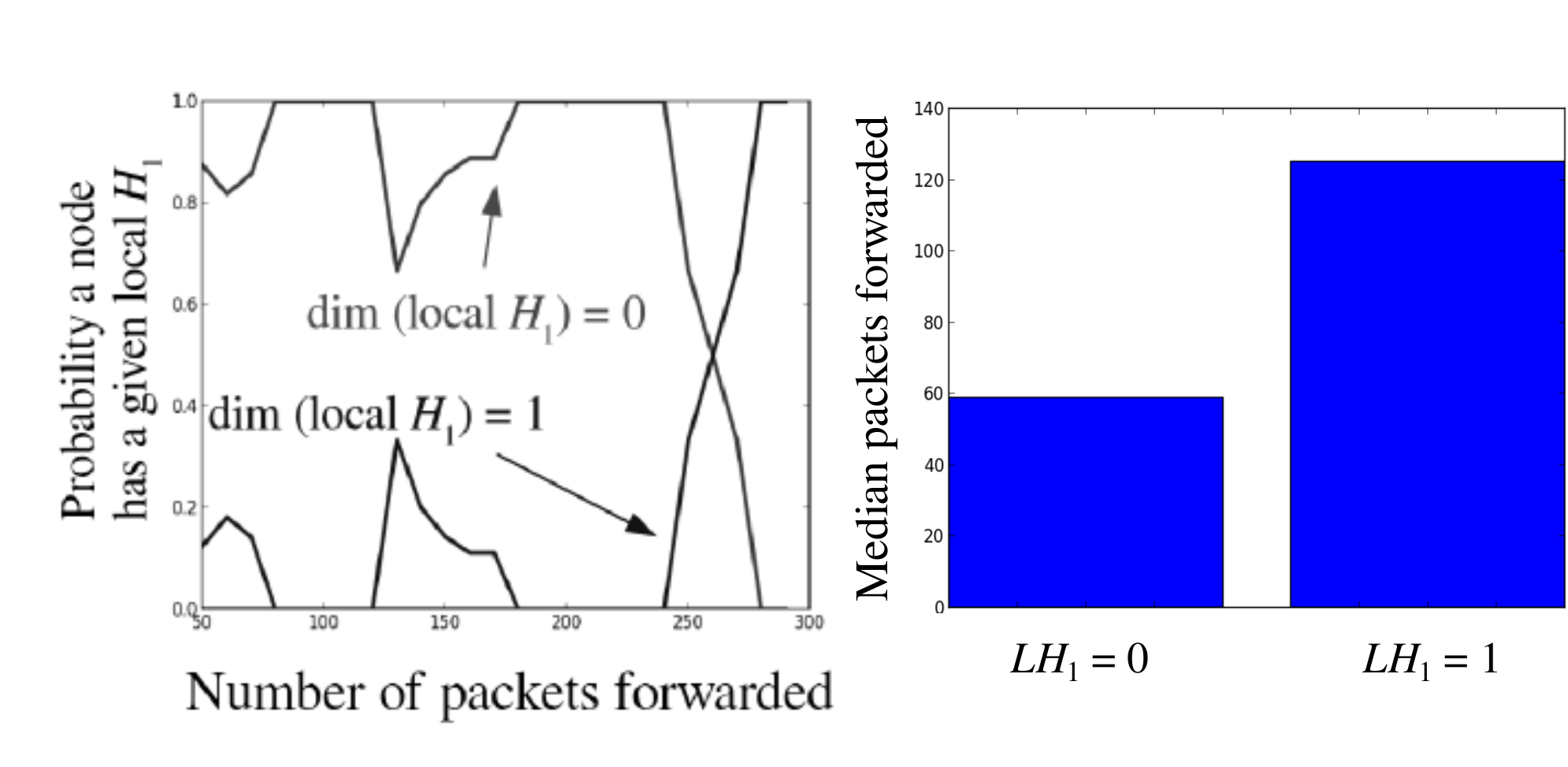} 
\caption{Probability a node has a certain local homology dimension given the number of packets it forwards}
\label{fig:fwdlh}
\end{center}
\end{figure}

Figure \ref{fig:fwdlh} shows the probability that a node forwarding a certain number of packets will have the given value of $LH_1$.  (We did not find a strong correspondence between forwarded packets and $LH_2$.)  It is immediately clear that all nodes forwarding a large number of packets are assigned a high local homology, but the converse is not necessarily true.  Local homology dimension is an indication that a node may be critical, but does not guaranteed that it actually is.

\section*{Acknowledgement}
This research was developed with funding from the Defense Advanced Research Projects Agency (DARPA) via Federal contract HR0011-15-C-0050 and the Air Force Office of Scientific Research via Federal contract FA9550-09-1-0643.  The views, opinions, and/or findings expressed are those of the authors and should not be interpreted as representing the official views or policies of the Department of Defense or the U.S. Government.

\bibliographystyle{IEEEtran}
\bibliography{activation_bib}

\section*{Appendix}
Although the space of global sections for an activation sheaf is a useful invariant, its sheaf cohomology is rather uninteresting.  We need to enrich their structure somewhat to see this, though.

\begin{df}
  If $\shf{A}$ is an activation sheaf on an abstract simplicial complex $X$, the \emph{vector activation sheaf} $\widehat{\shf{A}}$ is given by specifying its stalks and restrictions:
\begin{enumerate}
\item To each cell $c\in X$, let $\widehat{\shf{A}}(c)$ be the vector space whose basis is $\shf{A}\backslash\{\perp\}$ (so the dimension of this vector space is the cardinality of $\shf{A}$ without counting $\perp$)
\item The restriction map $\widehat{\shf{A}}(c\subset d)(n)$ is the basis projection, which is well-defined since $\shf{A}(d) \subseteq \shf{A}(c)$.
\end{enumerate}
\end{df}

\begin{thm}
  The dimension of the cohomology spaces of a vector activation sheaf $\widehat{\shf{A}}$ on a link complex $X$ are
  \begin{equation*}
    \text{dim }H^k(\widehat{\shf{A}}) = \begin{cases}
      \text{the total number of nodes}&\text{if }k = 0\\
      0&\text{otherwise}\\
      \end{cases}
  \end{equation*}
\end{thm}
\begin{proof}
  Every global section of $\shf{A}$ corresponds to a global section of $\widehat{\shf{A}}$, but formal linear combinations of global sections of $\shf{A}$ are also global sections of $\widehat{\shf{A}}$.  Therefore, a global section of $\widehat{\shf{A}}$ merely consists of a list of those nodes that are transmitting, without regard for whether they interfere.

  The fact that the other cohomology spaces are trivial is considerably more subtle.  Consider the decomposition
  \begin{equation*}
    X = \bigcup_i F_i
  \end{equation*}
  of the link complex into the set of its facets.  Suppose that $F_i$ is a facet of dimension $k$, and define $\shf{F}_i$ to be the direct sum of $k+1$ copies of the constant sheaf supported on $F_i$.  (Each copy corresponds one of the vertices of $F_i$.)  Then there is an exact sequence of sheaves
  \begin{equation*}
    \xymatrix{
      0 \to \widehat{\shf{A}} \ar[r]^\Delta & \bigoplus_i \shf{F}_i \ar[r]^m & \shf{S} \to 0\\
      }
  \end{equation*}
  where $\Delta$ is a map that takes a basis vector corresponding to a given node to the linear combination of all corresponding basis vectors in each copy of the constant sheaves, and $m$ is therefore a kind of difference map.  This exact sequence leads to a long exact sequence
  \begin{equation*}
\dotsb H^{k-1}(\shf{S}) \to H^k(\widehat{\shf{A}}) \to \bigoplus_i H^k(\shf{F}_i) \to H^k(\shf{S}) \dotsb
  \end{equation*}
Since each $\shf{F}_i$ is a direct sum of constant sheaves supported on a closed subcomplex, it only has nontrivial cohomology in degree 0.

Observe that $\shf{S}$ is a sheaf supported on sets of cells lying in the intersections of facets.  By Corollary \ref{cor:unaffected}, $\shf{S}$ must be a direct sum of copies of constant sheaves supported on closed subcomplexes, like each $\shf{F}_i$.  Thus $\shf{S}$ only has nontrivial cohomology in degree 0, which means that for $k>1$, $H^k(\widehat{\shf{A}}) = 0$.

It therefore remains to address the $k=1$ case, which comes about from the exact sequence
\begin{equation*}
  \bigoplus_i H^0(\shf{F}_i) \to H^0(\shf{S}) \to H^1(\widehat{\shf{S}}) \to 0.
\end{equation*}
The leftmost map is surjective, since every global section of $\shf{S}$ is given by specifying a single transmitting node.  By picking exactly one facet containing that node, a global section of the corresponding $\shf{F}_i$ may be selected in the preimage.  Thus the map $H^0(\shf{S}) \to H^1(\widehat{\shf{S}})$ must be the zero map and yet also surjective.  This completes the proof.
\end{proof}

\end{document}